\documentclass[11pt]{article} 
\usepackage{harvard,fancyheadings,bbm}
\usepackage{graphicx}
\usepackage[space]{grffile}
\usepackage{color}
\usepackage{fullpage}
\usepackage[svgnames]{xcolor}
\usepackage{tikz,pgfplots}
\usetikzlibrary{patterns}
\usepackage{array}
\usepackage{multirow}

\usepackage{amsmath,amsfonts}
\usepackage{amsthm}
\usepackage{amstext,rotating}
\usepackage{amsfonts,amssymb,graphics,xspace}
\usepackage{epsfig}
\numberwithin{equation}{section}

\newtheorem{assumption}{Assumption}[section]
\newtheorem{definition}{Definition}
\newtheorem{theorem}{Theorem}[section]
\newtheorem{lemma}{Lemma}[section]
\newtheorem{proposition}{Proposition}

\newcommand{\ba}{\begin{array}}
\newcommand{\ea}{\end{array}}
\newcommand{\bs}{\begin{align}\begin{split}\nonumber}
\newcommand{\bsnumber}{\begin{align}\begin{split}}
\newcommand{\es}{\end{split}\end{align}}

\renewcommand{\[}{\left[}

\newcommand{\bi}{\begin{itemize}}
\newcommand{\ei}{\end{itemize}}
\newcommand{\be}{\begin{enumerate}}
\newcommand{\ee}{\end{enumerate}}

\begin{document}

\title{Estimating Semi-parametric Panel Multinomial Choice\\ Models using Cyclic Monotonicity\thanks
{Emails: xshi@ssc.wisc.edu, mshum@caltech.edu, wsong22@wisc.edu.
We thank Khai Chiong, Federico Echenique, Bruce E. Hansen,  Jack R. Porter, and seminar audiences at
Johns Hopkins, Northwestern, NYU, UC Riverside, UNC, and the Xiamen/WISE Econometrics Conference in Honor of Takeshi Amemiya for useful comments.  Pengfei Sui and Jun Zhang provided excellent research assistance. Xiaoxia Shi acknowledges the financial support of the Wisconsin Alumni Research Foundation via the Graduate School Fall Competition Grant.
}}

\author{Xiaoxia Shi\\University of Wisconsin-Madison
\and
Matthew Shum\\Caltech
\and
Wei Song\\University of Wisconsin-Madison}
\maketitle
\begin{abstract}
This paper proposes a new semi-parametric identification and estimation approach to multinomial choice models in a panel data setting with individual fixed effects. Our approach is based on {\em cyclic monotonicity}, which is a defining feature of
the random utility framework underlying multinomial choice models. From the cyclic monotonicity property, we derive identifying inequalities  without requiring any shape restrictions for the distribution of the random utility shocks.  These inequalities point identify model parameters under straightforward assumptions on the covariates.   We propose a consistent estimator based on these inequalities. 

\medskip{}

\end{abstract}

\section{Introduction}
Consider a panel multinomial choice problem where agent $i$ chooses from $K+1$ options
(labelled $k=0,\ldots,K$). Choosing option $k$ in period $t$ gives the agent indirect utility
\begin{equation}\label{eq:panelspec}
A_i^k+\beta'X_{it}^k + \epsilon_{it}^k,
\end{equation}
where $X_{it}^k$ is a $d_x$-dimensional vector of observable covariates that has support ${\cal X}$, $\beta$ is the vector of weights for the covariates in the agent's utility, $\vec{A_i}= (A^0_i,\ldots,A^K_i)'$ are agent-specific fixed effects, and $\epsilon_{it}^k$ are unobservable utility shocks the distribution of which is not specified. 
The agent chooses the option that gives her the highest utility: 
\begin{equation}\label{eq:panelchoice}
Y_{it}^k = 1\{\beta'X_{it}^k + A_{i}^k +\epsilon_{it}^k\geq \beta'X_{it}^{k'}+A_{i}^{k'}+\epsilon_{it}^{k'};\ \forall k'\},
\end{equation}
where $Y_{it}^k$ denotes the multinomial choice indicator. Let the panel data have the standard structure, that is, identically and independently distributed (i.i.d.) across $i$ and stationary across $t$. As is standard, normalize $\|\beta\| = 1$, $X_{it}^0 = \mathbf{0}_{d_x}$ and $A_i^0 = 0=\epsilon_{it}^0$. We will not impose location normalization for $\epsilon_{it}^k$ or $A_i^k$, and as a result, it is without loss of generality to assume that $X_{it}^k$ does not contain a constant. 

In this paper, we propose a new semi-parametric approach to the identification and estimation of $\beta$.  We exploit the notion of {\em cyclic monotonicity}, which is an appropriate generalization of ``monotonicity'' to multivariate (i.e. vector-valued) functions.
We first show that the cyclic monotonicity property applies to the vector of choice probabilities $\left\{ P(Y^k=1|X^0,\dots, X^K)\right\}_{k=0,1,\ldots,K}$ emerging from any multinomial choice model (including both panel as well as simpler cross-sectional models), when viewed as a function of the vector of linear utility indices $(\beta'X^0,\dots,\beta'X^K)'$. 

Applied to panel models, the cyclic monotonicity property implies a collection of moment inequalities in which the fixed effects are differenced out. We then give a necessary and sufficient condition for the point identification of $\beta$ based on these inequalities.  Two sets of sufficient primitive conditions are subsequently discussed. Notably, one of the two sets  of primitive conditions allows all regressors to be bounded. We finally propose a consistent estimator for $\beta$, the computation of which requires only convex optimization.

To our knowledge, this is the first paper that deals with the incidental parameter problem while achieving point identification for semi-parametric panel multinomial choice models.   For partial identification in these models, Pakes and Porter (2015) propose an alternative approach. Pakes and Porter  construct inequality restrictions that partially identify $\beta$ using an idea that can be viewed as generalizing  Manski's (1987) maximum score estimator for panel binary choice models to the panel multinomial setting. By comparison, this paper can be seen as generalizing Han's (1987, 1988) maximum rank-correlation estimator, which applies to cross-sectional binary choice mdoels, to the panel multinomial setting.\footnote{Abrevaya (1999) proposes a maximum rank-correlation estimator for panel transformation models. His approach does not apply to discrete choice models due to a strict monotonicity requirement on the transformation function.}
 
The literature on semi-parametric panel binary choice models is large.  Manski (1987) and Honor\'{e} and Kyriazidou (2000) use the maximum score approach, while Honor\'{e} and Lewbel (2002) generalize the special regressor approach of Lewbel (1998, 2000) to the panel data setting.  Identification conditions in these three papers are non-nested with ours. Chamberlain (2010) shows the impossibility of point identification in a binary choice special case of the model described by Eqs. (\ref{eq:panelspec}) and (\ref{eq:panelchoice}) when $X_{it}$ is bounded and contains a time dummy. This impossibility result is implied by our necessity result (Theorem B.1) which shows that uniform point identification is impossible if all regressors are bounded and at least one of them is finite-valued (e.g. the time dummy). When no regressor is finite-valued, boundedness does not preclude point identification, as shown in one of our sufficiency results (Theorem 3.2).
 
Semi-parametric identification and estimation of multinomial choice models have been considered in cross-sectional settings (i.e., models without individual fixed effect).  Manski (1975) and Fox (2007) base identification on the assumption of a {\em rank-order property} that the ranking of $\beta'X_i^k$ across $k$ is the same as that of  $E[Y_i^k|X_i]$ across $k$; this is an IIA-like property that allows utility comparisons among all the options in the choice set to be decomposed into pairwise comparisons among these options. To ensure this rank-order property, Manski assumes that the error terms are i.i.d. across $k$, while Fox relaxes the i.i.d. assumption to exchangeability. Exchangeability (or the rank-order property) is not used in our approach. In addition, Powell and Ruud (2008) and Ahn, Ichimura, Powell, and Ruud (2015) consider an alternative approach based on matching individuals with equal conditional choice probabilities, which requires that the rank of a certain matrix formed from the data to be deficient by exactly 1.  This approach does not obviously extend to the panel data setting with fixed effects.
 
 The existing literatures on cross-sectional binary choice models  and on the semi-parametric estimation of single or multiple index models (which include discrete choice  models as examples) is voluminous and less relevant for us, and thus is not reviewed here for brevity.\footnote{
An exhaustive survey is provided in Horowitz (2009), chapters 2 and 3.}

The paper proceeds as follows.
In section 2, we  introduce the notion of cyclic monotonicity and relate it to panel multinomial choice models with fixed effects.   Subsequently, in Section 3, we present the moment inequalities emerging from cyclic monotonicity, and give assumptions under which these inequalities suffice to point identify the parameters of interest.  This section also contains some numerical illustrations.
  Section 4 presents an estimator, shows its consistency, and evaluates its performance using Monte Carlo experiments. In Section 5, we discuss the closely related aggregate panel multinomial choice model, which is a workhorse model for demand modelling in empirical IO.   This section also contains an illustrative empirical application using aggregate supermarket scanner data.  Section 6 concludes.


\section{Cyclic Monotonicity and Multinomial Choice Models}
We begin by defining cyclic monotonicity, the central notion of this paper.

\begin{definition}[Cyclic Monotonicity] Consider a function $f:{\cal U}\to R^K$ where ${\cal U}\subseteq R^K$, and a length $M$-cycle of points in $R^K$: $u_1,u_2,\dots,u_M, u_1$. The function $f$ is cyclic monotone  with respect to the cycle $u_1,u_2,\dots,u_M, u_1$ if and only if
\begin{equation}
\sum_{m=1}^M (u_{m}-u_{m+1})'f(u_m) \geq 0,
\end{equation}
where $u_{M+1}=u_1$. The function $f$ is cyclic monotone on ${\cal U}$ if it is cyclic monotone with respect to all possible cycles of all lengths on its domain.\footnote{Technically, this defines  the property of being ``cyclic monotonically increasing,'' but for notational simplicity and without loss of generality, we use ``cyclic monotone'' for ``cyclic monotonically increasing.''} 
\end{definition}

For real-valued functions defined on a real-space (i.e., $K=1$), cyclic monotonicity is equivalent to monotonicity.  In this sense, cyclic monotonicity generalizes monotonicity in a vector-valued context.
We make use of the following basic result which relates cyclic monotonicity to convex functions:
\begin{proposition}[Cyclic monotonicity and Convexity]\label{prop:Convexity} Consider a differentiable function $F:{\cal U}\to R$ for an open convex set\ ${\cal U}\subseteq R^K$. If $F$ is convex on ${\cal U}$, then the gradient of $F$ (denoted $\nabla F(u):=\partial F(u)/\partial u$) is cyclic monotone on ${\cal U}$. 
\end{proposition}
The proof for Proposition \ref{prop:Convexity} is available from standard sources  (e.g, Rockafellar (1970, Ch. 24), Villani (2003, Sct. 2.3)).
Consider a univariate and differentiable convex function; obviously, its slope must be monotonically nondecreasing.   The above result states that cyclic monotonicity is the appropriate extension of this feature to multivariate convex functions.

\medskip

Now we connect the above discussion to the multinomial choice model. We start with a generic random utility model for multinomial choices without specifying the random utility function or the data structure in detail. Suppose that an agent is choosing from $K+1$ choices $0, 1, \dots, K$. The utility that she derives from choice $k$ is partitioned into two additive parts: $U^k+\epsilon^k$, where $U^k$ denotes the systematic component of the latent utility, while $\epsilon^k$ denotes the random shocks, idiosyncratic across agents and choice occasions. She chooses choice $k^\ast$ if $U^{k^\ast} +\epsilon^{k^\ast}\geq \max_{k=0,\dots,K}U^k+\epsilon^k$. Let $Y^k = 1$ if she chooses choice $k$ and $0$ otherwise. As is standard, we normalize $U^0 = \epsilon^0 =0$. 

Let $u^k$ denote a generic realization of $U^k$. Also let $\vec{U} = (U^1,\dots,U^K)'$, $\vec{u}=(u^1,\dots,u^K)'$, and $\vec{\epsilon} = (\epsilon^1,\dots,\epsilon^K)'$. Then we can define a function that is a stepping stone for applying cyclic monotonicity in the multinomial choice context. The function, which McFadden (1978, 1981) called the ``social surplus function,'' is the expected utility obtained from the choice problem:
\begin{equation}
{\cal G}(\vec{u}) = E\left\{\max_{k=0,\dots,K}[U^k+\epsilon^k]|\vec{U}=\vec{u}\right\}.
\end{equation}
The following lemma shows that this function is convex,  that the gradient of it is the choice probability function, and finally that the choice probability function is cyclic mononotone. 

\begin{lemma}[Gradient]\label{lem:gradient} Suppose that $\vec{U}$ is independent of $\vec{\epsilon}$ and that the cumulative distribution function (c.d.f.) of $\vec{\epsilon}$ is continuous everywhere. Then 

\textup{(a)} ${\cal G}(\cdot)$ is convex on $R^{K}$,

\textup{(b)} ${\cal G}(\cdot)$ is differentiable on $R^{K}$,

\textup{(c)} $\vec{p}(\vec{u}) = \nabla{\cal G}(\vec{u})$, where $\vec{p}(\vec{u}) = E[\vec{Y}|\vec{U} = \vec{u}]$ and $\vec{Y} = (Y^1,\dots, Y^K)'$, and

\textup{(d)} $\vec{p}(\vec{u}) $ is cyclic monotone on $R^{K}$.
\end{lemma}

The cyclic monotonicity of the choice probability can be used to identify the structural parameters in $\vec{U}$ in a variety of  settings. In this paper, we focus on the linear panel data model with fixed effects, composed of equations (\ref{eq:panelspec}) and (\ref{eq:panelchoice}). 

\section{Panel Data Multinomial Choice Models with Fixed Effects}

Consider a short panel data setting where there are $T$ time periods. Let $\vec{U}$, $\vec{\epsilon}$, and $\vec{Y}$ be indexed by both $i$ (individual) and $t$ (time period). Thus they are now $\vec{U}_{it} \equiv (U_{it}^1, \dots,U_{it}^K)'$, $\vec{\epsilon}_{it} \equiv (\epsilon_{it}^1, \dots,\epsilon_{it}^K)'$, and $\vec{Y}_{it} \equiv (Y_{it}^1, \dots,Y_{it}^K)'$. Let there be an observable $d_x$ dimensional covariate $X^k_{it}$ for each choice $k$ and let ${U}^k_{it}$ be a linear index of $X^k_{it} $ plus an unobservable individual effect $A^k_i$:
\begin{align}
U^k_{it} = \beta'X^k_{it}+A^k_{i},
\end{align}
where $\beta$ is a $d_x$-dimensional unknown parameter. Let $\vec{X}_{it} = (X_{it}^1,\dots,X_{it}^K)$ and $\vec{A}_i = (A^1_i, \dots, A^K_i)'$. Note that $\vec{X}_{it}$ is a $d_x\times K$ matrix. In short panels, the challenge in this model is the identification of $\beta$ while allowing correlation between the covariates and the individual effects. We tackle this problem using the cyclic monotonicity of the choice probability, as we explain next.

 \subsection{Identifying Inequalities}
We derive our identification inequalities under the following assumption.

\begin{assumption}\label{ass:error}
\textup{(a)} The error term $\vec{\epsilon}_{it}$ is independent of $\vec{X}_{it}$ given $\vec{A}_i$, 

\textup{(b)} $E[\vec{Y}_{it}|\vec{X}_{i1},\dots, \vec{X}_{iT}, \vec{A}_i] $ $= E[\vec{Y}_{it}|\vec{X}_{it},\vec{A}_i]$ for all $t$, and 

\textup{(c)} the conditional c.d.f. of $\vec{\epsilon}_{it}$ given $\vec{A}_i$ is continuous everywhere.
\end{assumption}
\noindent \textbf{Remark.} As we see in the derivation below, Assumption \ref{ass:error}(a) ensures cyclic monotonicity, while Assumption \ref{ass:error}(b) allows us to integrate out the unobservable individual effect. In terms of dependence between the covariates and the errors, these conditions require only {\em conditional} independence given $\vec{A}_i$.   (They accommodate, for example, heteroskedasticity that depends on $\vec{A}_i$.) 
Furthermore, they can also be easily weakened when there are reasonable control variables, even when those control variables are unobservable and are of  infinite dimension.  We discuss these extensions in Appendix \ref{sec:control}.

Moreover, in terms of dependence amongst the errors, i.i.d. errors (whether across time periods or across choices) are not required. The errors across choices within one time period can have arbitrary joint distribution; the errors across time periods can be arbitrarily dependent.\hfill$\blacksquare$

Under Assumptions \ref{ass:error}(a) and (c), Lemma \ref{lem:gradient} implies that  the conditional choice probability
 \begin{align}
\vec{p}(\vec{v},\vec{a}) = E[\vec{Y}_{it}|\vec{X}_{it}'\beta= \vec{v}, \vec{A}_i = \vec{a}] \label{eq:pva}\end{align} 
is cyclic monotone in $\vec{v}$ for any fixed $\vec{a}$.  Let $\vec{\cal X}$ be the support of $\vec{X}_{it}$. The cyclic monotonicity of $\vec{p}(\vec{v},\vec{a})$ with respect to $\vec{v}$ for any fixed $\vec{a}$ can be written as
\begin{align}
\sum_{m=1}^M p(\beta'\vec{x}_m,\vec{a})'(\vec{x}_{m}'\beta - \vec{x}_{m+1}'\beta)\geq 0, \forall \vec{a}, ~\forall \vec{x}_1,\vec{x}_2,\dots,\vec{x}_M\in \vec{\cal X} \text{ and } \forall M, \label{cmpm}
\end{align}
where $\vec{x}_{M+1} = \vec{x}_1$. These inequalities cannot be used directly to identify $\beta$ because the conditional choice probability function $p(\vec{v}, \vec{a})$ is not identified due to the latency of $\vec{A}_i$. However, we show that $\vec{a}$ can be integrated out with panel data.

Now use Assumption \ref{ass:error}(b), and we have that, for any cycle $t_1,t_2,\dots,t_M, t_{M+1}=t_1$ in  $\{1,\dots, T\}$, and any $(\vec{x}_1,\dots,\vec{x}_M, \vec{a})$ in the support of $(\vec{X}_{it_1},\dots,\vec{X}_{i,t_M},\vec{A}_i)$, 
\begin{align}
p(\beta'\vec{x}_m,\vec{a})& = E[\vec{Y}_{it_m}|\vec{X}_{it_m}=\vec{x}_m,\vec{A}_i=\vec{a} ]\nonumber\\
& = E[\vec{Y}_{it_m}|\vec{X}_{it_1}=\vec{x}_1,\dots,\vec{X}_{it_M}=\vec{x}_M,\vec{A}_i=\vec{a}]~\forall m=1,\dots,M.  
\label{eq34}\end{align}
This and eq. (\ref{cmpm}) together implies that, for any positive integer $M\leq T$ and any cycle $t_1,t_2,\dots,t_M,$ $t_{M+1}=t_1$ in  $\{1,\dots, T\}$,
\begin{align}
\sum_{m=1}^ME[\vec{Y}_{it_m}'|\vec{X}_{it_1},\dots,\vec{X}_{it_M},\vec{A}_i](\vec{X}_{it_{m}}'\beta - \vec{X}_{it_{m+1}}'\beta)\geq 0 \text{ almost surely}. \label{preintegratem}
\end{align}
Take conditional expectation given $\vec{X}_{it_1}, \dots,\vec{X}_{it_M}$ of both sides, and we get, for any positive integer $M\leq T$ and any cycle $t_1,t_2,\dots,t_M,$ $t_{M+1}=t_1$ in  $\{1,\dots, T\}$,
\begin{align}
\sum_{m=1}^ME[\vec{Y}_{it_m}'|\vec{X}_{it_1},\dots,\vec{X}_{it_M}](\vec{X}_{it_{m}}'\beta - \vec{X}_{it_{m+1}}'\beta)\geq 0 \text{ almost surely}. \label{panelidm}
\end{align}
These inequality restrictions involve only identified/observed quantities and the unknown parameter $\beta$, and thus can be used to identify $\beta$.

We summarize the result of the derivation in a lemma below. The proof for the lemma has already been given above in the discussion around Eqs. (\ref{eq:pva})-(\ref{panelidm}).

\begin{lemma}\label{lem:idineq} Under Assumption \textup{\ref{ass:error}}, 
$$\sum_{m=1}^ME[\vec{Y}_{it_m}'|\vec{X}_{it_1},\dots,\vec{X}_{it_M}](\vec{X}_{it_{m}}'\beta - \vec{X}_{it_{m+1}}'\beta)\geq 0 \text{ almost surely}.$$
\end{lemma}

\subsection{Point Identification of Model Parameters}

To see the amount of identification information the inequalities in (\ref{panelidm}) contain, write (\ref{panelidm}) as
\begin{equation}
\beta'g_{t_1,\dots,t_M}(\vec{X}_{it_1},\dots,\vec{X}_{it_M})\geq 0,~ \forall t_1,t_2,\dots,t_M, t_{M+1}=t_{1}\in\{1,\dots,T\}, ~\forall M\leq T\text{, almost surely},\label{panelid2}
\end{equation}
where $g_{t_1,\dots,t_M}(\vec{x}_1,\dots,\vec{x}_{M}) = \sum_{m=1}^M\{(\vec{x}_{m}-\vec{x}_{m+1})E[\vec{Y}_{it_m}|\vec{X}_{it_1}=\vec{x}_1,\dots,\vec{X}_{it_M}=\vec{x}_M]\}$. Let ${\cal G}_{t_1,\dots,t_M} $ be the support of $g_{t_1,\dots,t_M}(\vec{X}_{it_1},\dots,\vec{X}_{it_M})$. Let
 \begin{equation}
 {\cal G}_M = \cup_{t_1,\dots,t_M\in\{1,\dots,T\}} ~{\cal G}_{t_1,\dots,t_M}
\end{equation}
Let 
\begin{equation}
{\cal G} = \cup_{M=2,\dots,T} ~{\cal G}_{M}.
\end{equation}
Then the identified set (denoted by $B_0$) of $\beta$ defined by the restriction (\ref{panelidm}) is the set 
\begin{equation}
B_0 = \{b\in R^{d_x}: \|b\| = 1, b'g\geq 0, ~\forall g\in {\cal G}\}.\label{idset}
\end{equation}
This set is a proper subset of $\{b\in R^{d_x}:\|b\|=1\}$ as long as ${\cal G}$ contains at least one nonzero element. 

In order to shrink $B_0$, the set ${\cal G}$ must grow richer. In fact ${\cal G}$ must grow in such a way that $cc({\cal G})$ grows bigger, where $cc({\cal G})$ is the closed convex cone generated by ${\cal G}$:
\begin{align}
cc({\cal G}) = closure(\{\lambda_1g_1+\lambda_2g_2: \lambda_1\geq 0, \lambda_2\geq 0, g_1, g_2\in {\cal G}\}).
\end{align}  
This is because, by elementary algebra, $B_0$ can be written as
\begin{align}
B_0 = \{b\in R^{d_x}: \|b\| = 1, b'g\geq 0, ~\forall g\in cc({\cal G})\}.\label{sharp}
\end{align}
When $cc({\cal G})$ is so rich that it becomes a half-space of $R^{d_x}$, $\beta$ is point identified, as shown in the following theorem, the proof of which is in Appendix \ref{proofs}.

\begin{theorem}\label{thm:setid} 
 The identified set $B_0 = \{\beta\}$ if and only if $cc({\cal G}) $ is a half-space of $R^{d_x}$.
\end{theorem}

Figure \ref{identpic} illustrates the identification argument, for the case when the vectors $\beta$ and $g$ lie in the Cartesian plane, and can be represented as points on the unit circle.  In this case, the identified set (\ref{idset}) can be visualized as the collection of norm 1 vectors that are acute with respect to all the vectors in ${\cal G}$.  In Panel (i), we show the worst case scenario where the set $cc({\cal G})$ consists of a single vector, given by the solid green arrow.   For this case, the identified set consists of the entire halfspace or halfcircle $OACB$, which are all the vectors which form an acute angle with the vector in $cc({\cal G})$.   

\begin{figure}\label{identpic}
\vspace{-2in}
$
\begin{array}{ccc}
\hspace{-1.75in}\includegraphics[width=0.75\linewidth]{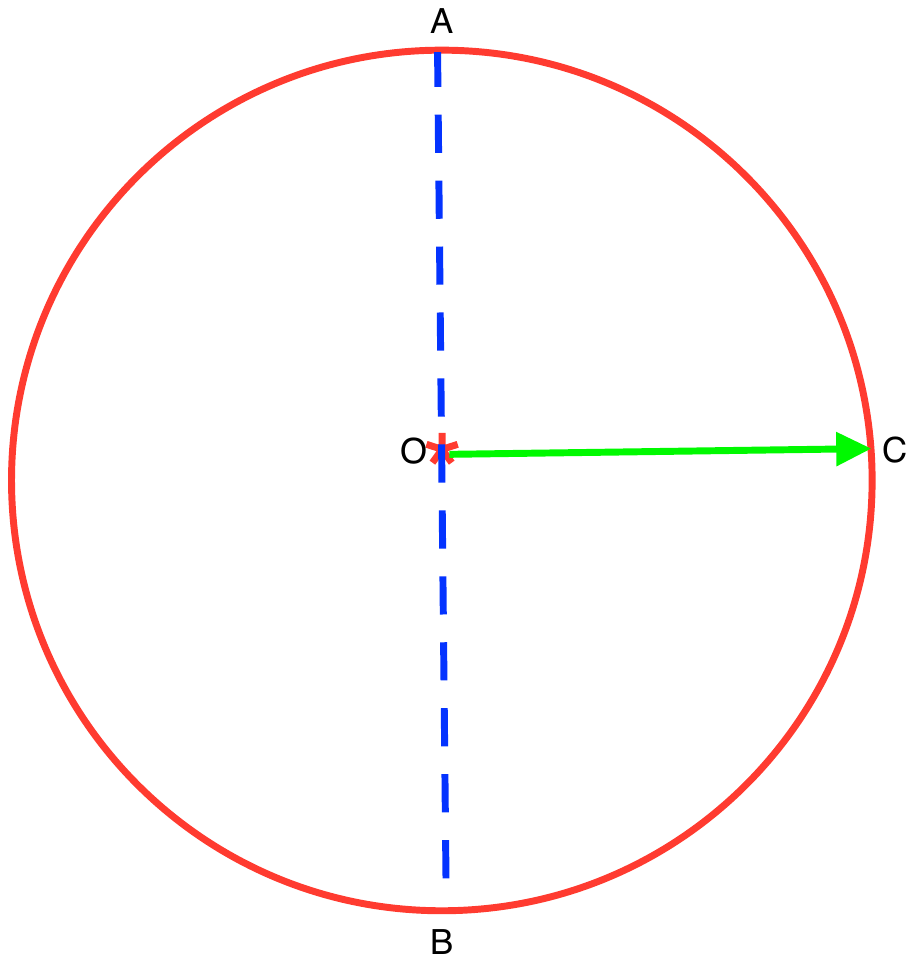}&
\hspace{-2.5in}\includegraphics[width=0.75\linewidth]{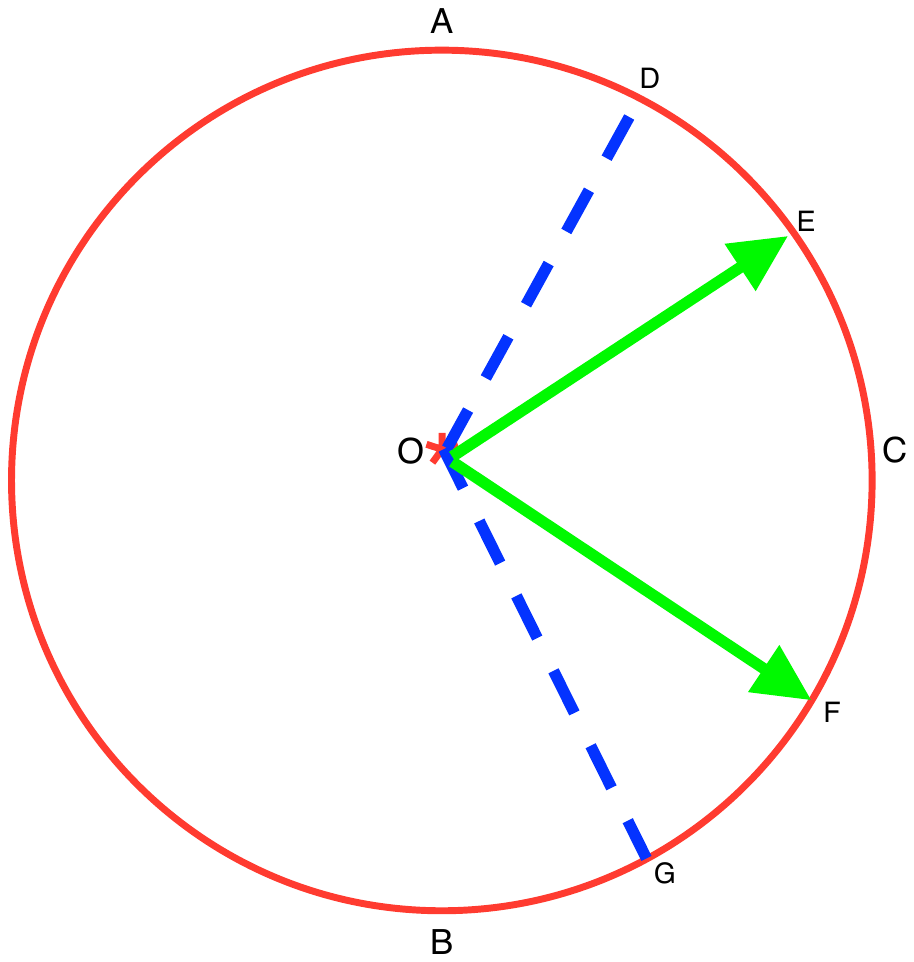}&
\hspace{-2.5in}\includegraphics[width=0.75\linewidth]{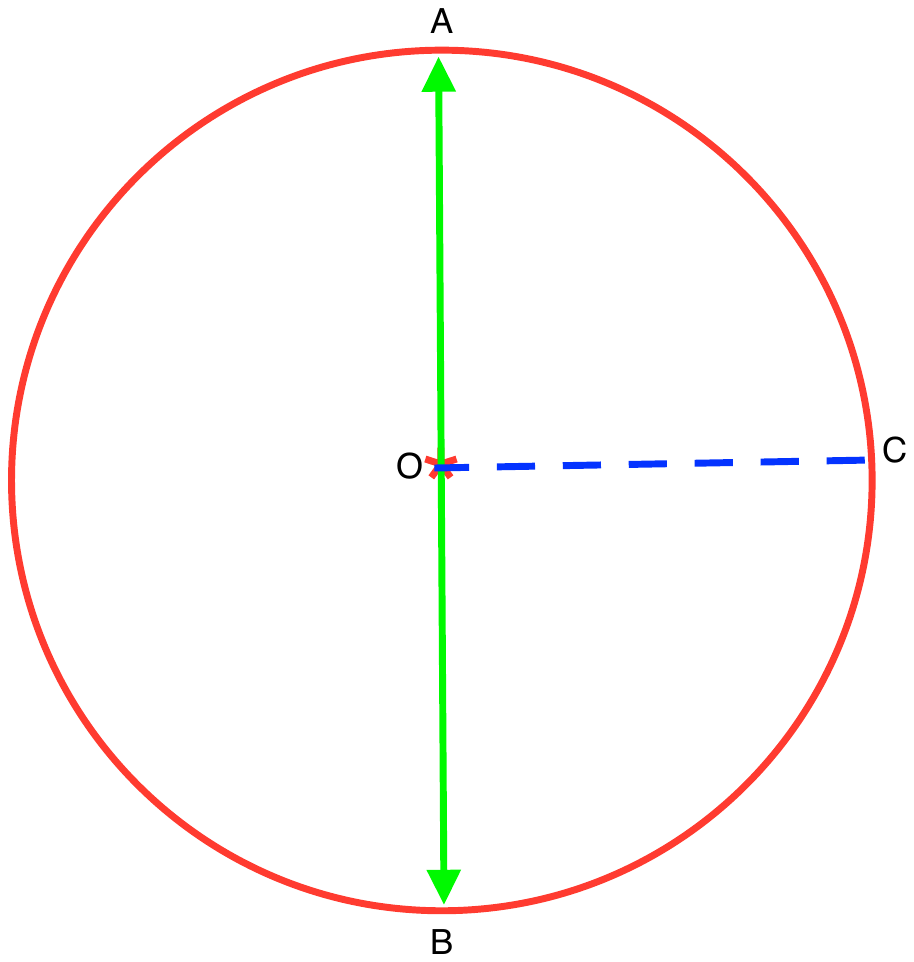}\\[-2in]
\hspace{-1.5in}\text{Panel (i)}&
\hspace{-2.5in}\text{Panel (ii)}&
\hspace{-2.5in}\text{Panel (iii)}\\
\end{array}
$
\caption{Point identification of $\beta$: geometric intuition}

{\footnotesize In all the above panels, the green arrows delineate the set $cc({\cal G})$ (defined in the text), and the blue dotted lines delineate the identified set for $\beta$.   Moving from Panel (i) to (iii), we see how the identified set shrinks as $cc({\cal G})$ grows.}
\end{figure}

The remaining panels show how the identified set shrinks as $cc({\cal G})$ becomes richer.   In panel (ii), $cc({\cal G})$ expands to the slice $OEF$ (bounded by the solid green arrows), which shrinks the identified set to $ODG$ (bounded by the dotted blue lines), which as before are the vectors that are acute with every vector in $cc({\cal G})$.   Finally, in panel (iii), we show the case of point identification: as $CC({\cal G})$ grows to become the entire halfspace/halfcircle $OACB$, the identified set shrinks down to the single vector $OC$, given by the blue dotted line.

\subsection{Primitive Point Identification Conditions}
Theorem \ref{thm:setid} gives a necessary and sufficient condition for the point identification of $\beta$ under the assumptions of Lemma \ref{lem:idineq}. The condition is based on the support of observables, and thus is in principle verifiable given an infinite amount of data. In finite samples, however, testing support richness is difficult if at all possible. Moreover,  it is difficult to logically argue for/against the condition or to compare it  to what is available in the literature because it involves the non-primitive components $E[\vec{Y}_{it}|\vec{X}_{it}]$. 

Here, we introduce conditions that are based on the model primitives $\vec{X}_{it}$ and $\vec{\epsilon}_{it}$. 
We focus on identification based on only the the length-2 cycles.  First, point identification using only the length-2 cycles implies point identification using more or all cycles because using more cycles adds restrictions.  Second, in practice, estimation with longer cycles not only is computationally more intensive, but also requires estimating a higher-dimensional conditional choice probability function ($E[\vec{Y}_{it_m}|\vec{X}_{it_1}, \vec{X}_{it_2},\dots,\vec{X}_{it_M}])$. One may be constrained to use only the length-2 cycles. Thus, identification based on only the length-2 cycles is arguably most practically useful. 

The following notation will be used. For each $k=1,\dots, K$, let $\vec{X}_{it}^{-k} = (X_{it}^1,\dots,X_{it}^{k-1},\allowbreak X_{it}^{k+1},\dots,X_{it}^K)$ and let $\vec{\cal X}_{-k}$ denote the support of $\vec{X}_{it}^{-k}$. For  a generic element $\vec{x}^{-k}$  in $\vec{X}_{it}^{-k}$, let $G_{s,t}^{k}(\vec{x}^{-k})$ be the conditional support of $X_{it}^k - X_{is}^k$ given that $\vec{X}_{it}^{-k} =\vec{X}_{is}^{-k}= \vec{x}^{-k}$. Let 
\begin{equation}
G = \cup_{s,t=1,\dots,T}\cup_{k=1,\dots,K}\cup_{\vec{x}^{-k}\in \vec{\cal X}_{-k}} G_{s,t}^k(\vec{x}^{-k}).\label{G}
\end{equation}
Also define the {\em cone} generated by $G$ as
\begin{equation}
cone(G) =\{\lambda g\in R^{d_x}: \lambda\geq 0, g\in G\}.
\end{equation}
Theorem \ref{thm:sufbdm} characterizes the first set of primitive sufficient conditions using the set $G$ defined in eq. (\ref{G}). 
\begin{assumption}\label{ass:error2m}
For every $k=1,\dots, K$, with positive probability, the conditional support of $\vec{\epsilon}_{it}$ given $\vec{A}_i$ is $R^K$. 
\end{assumption}
\begin{assumption}\label{ass:sufbdm}
The set $cone(G)$ is dense in $R^{d_x}$. 
\end{assumption}
\begin{theorem}[{\bfseries Sufficient Point Identification Conditions - Set A}]\label{thm:sufbdm} Under Assumptions \textup{\ref{ass:error}(a)-(b)} and \textup{\ref{ass:error2m}}, if Assumption \textup{\ref{ass:sufbdm}} holds,  then $cc({\cal G}_M)$ for $M=2$ is a half-space, and so is $cc({\cal G})$.
\end{theorem}

The most interesting feature of Theorem \ref{thm:sufbdm} is that it establishes point identification allowing bounded, even discrete, support for all regressors. The following two examples illustrate this feature.\medskip

\noindent\textbf{Example 1 (Bounded Regressors).} Suppose that for two time periods, $t=t_1, t_2$, and some $k=1,\dots,K$, given the event that $\vec{X}_{it}^{-k}$ does not change across the two time periods, $X_{it}^k$ can change by any amount on the hypercube $[-c,c]^{d_x}$ for some $c$ (no matter how small), then $G$ contains   $[-c,c]^{d_x}$, and $cone(G) = R^{d_x}$. Note that the covariates that are held fixed, i.e., $\vec{X}^{-k}_{it}$, can be finite-valued. 
\medskip

\noindent\textbf{Example 2 (Discrete Regressors).} Suppose that $d_x=2$, and  that, for two time periods $t,s$ and some $k=1,\dots,K$, given the event that $\vec{X}_{it}^{-k}$ does not change across the two time periods, the support of $(X_{it}^{k,\prime}, X_{is}^{k,\prime})' $ is $\{0,1, 1/2, 1/3, 1/4, \dots\}^4$. Then it can be verified that $\{g_1/g_2:(g_1,g_2)'\in G ~s.t.~g_2\neq 0\}$ contains the set of all rational numbers, which implies that $cone(G)$, being a superset of the union of all rays with rational directions in $R^2$, is dense in $R^2$.  The same conclusion can be drawn when the support of $(X_{it}^{k,\prime}, X_{is}^{k,\prime})' $ is $\{1,2,3,\dots\}^4$, too. Like in the previous example, the covariates that are held fixed can be finite-valued.\footnote{One may argue that ``true'' discrete variables do not take a countably infinite number of values. But if it takes a reasonably large number of values, this example may be considered a good theoretical approximation. In Section \ref{sec:num} below, we use a numerical example to illustrate how fast point identification is approached as we add support points to discrete random variables.}  
\medskip

Theorem \ref{thm:sufbdm} allows some finite-valued regressors as discussed in the examples above. However, it does not allow, for example, that for some $j, j'=1,\dots, d_x$ and $j\neq j'$, $X_{j, it}^k$ and $X_{j',it}^k$ are finite-valued for all $k=1,\dots, K$. In that case, the projection of $G$ onto its $(j,j')$th coordinates is finite-valued and cannot generate a cone that is dense in $R^2$. Thus, $G$ cannot generate a cone dense in $R^{d_x}$ either. Next, we present a different set of sufficient conditions that allows this case at the expense of requiring a regressor with large support. 

For $j=1,\dots, d_x$, let $g_j$ denote the $j$th element of the $d_x$-dimensional vector $g$, and let $g_{-j} = (g_1,\dots, g_{j-1},g_{j+1},\dots, g_{d_x})'$. Let $G_{-1} = \{g_{-1}\in R^{d_x-1}: \exists g_1\in R ~s.t.~ (g_1,g_{-1}')'\in g\}$. For any $g_{-1}\in G_{-1}$, let $G_1(g_{-1}) = \{g_1\in R: (g_1,g_{-1}')' \in G\}$.  Similarly define $G_{-j}$ and $G_j(g_{-j})$ for $j=2,\dots,d_x$. 

\begin{assumption}\label{ass:suflgm}
For some $j\in\{1,\dots,d_x\}$, 

\textup{(a)} $G_j(g_{-j}) = R$ for all $g_{-j}$ in  $G_{-j}$, 

\textup{(b)} $G_{-j}$ is not contained in a proper linear subspace of $R^{d_x-1}$,  and 

\textup{(c)} the $j$th element of $\beta$, denoted by $\beta_j$, is nonzero.
\end{assumption}

\begin{theorem}[{\bfseries Sufficient Point Identification Conditions - Set B}]\label{thm:suflgm} Under Assumptions \textup{\ref{ass:error}(a)-(b)} and \textup{\ref{ass:error2m}}, if Assumption \textup{\ref{ass:suflgm}} holds,  then $cc({\cal G}_M)$ for $M=2$ is a half-space in $R^{d_x}$ and so is $cc({\cal G})$. 
\end{theorem}

Assumption \ref{ass:suflgm} is reminiscent of the covariate conditions in Manski (1987) for panel binary choice models with fixed effect (and also of the identification conditions in Manski (1975, 1988) and Han (1987) for cross-sectional binary choice models). Assumption \ref{ass:suflgm} is clearly different as it applies to the general multinomial choice case. There is still some difference even when we specialize to the binary choice case, which we discuss next.

In a two-period panel setting, Manski (1987) requires the support of one non-redundant element (say, $j$) of $X_{i1}^1 - X_{i2}^1$ to be $R$ conditional on the other elements. Assumption \ref{ass:suflgm}(a) is weaker in that it only requires this conditional support to contain either $[0,\infty)$ or $(-\infty,0]$ because $G_j(g_{-j})$ is the union of the conditional support of $X_{j,i1}^1 - X_{j, i2}^1$ and $X_{j,i2}^1 - X_{j,i1}^1$. Such a difference can be meaningful when, for example, $X_{j,it}^1$ can only grow over time.


Assumption \ref{ass:sufbdm} does not require any regressor to have
unbounded support, but it only allows discreteness to a limited extent. On
the other hand, Assumption \ref{ass:suflgm} allows almost all
regressors to be discrete (with finite support), but requires at least
one regressor with unbounded support. Comparing the two sets of
conditions, one notice a tradeoff between large support of one
regressor and rich support of all other regressors. In some sense, such a tradeoff is necessary.  In Appendix \ref{primnec}, we show that, for the special case of binary choice, when there is a finite-valued regressor, it is a necessary condition for point identification that at least {\em some} of the other regressors have unbounded support.   Because we were only able to show this for the binary choice case, we put that result in the appendix.

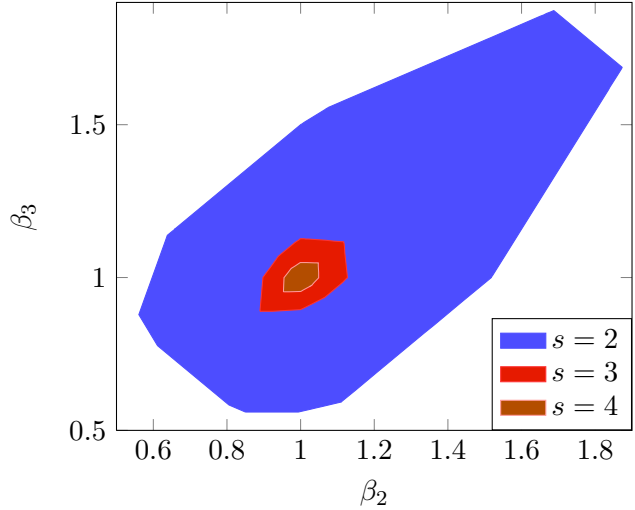
\begin{figure}\label{fig:tri}
\centering
\begin{tikzpicture}[scale = 1]
\begin{axis}[
xlabel = $\beta_2$,
ylabel = $\beta_3$,
xmin = 0.5,
xmax = 1.9,
ymin = 0.5,
ymax = 1.9,
legend style = {at={(1,0)}, anchor = south east},
]
\addplot [fill =blue!70!white,postaction={pattern = crosshatch},draw =blue!70!white,area legend] table[col sep = comma, x = x, y = y]{p2ch.txt};

\addplot [fill =red!90!green,postaction={pattern = grid},draw =red!80!white,area legend] table[col sep = comma, x = x, y = y]{p3ch.txt};

\addplot [fill =red!70!green,postaction = {pattern = north east lines},draw = red!40!white,area legend] table[col sep = comma, x = x, y = y]{p4ch.txt};

\legend{$s=2$,$s=3$, $s=4$}
\end{axis}
\end{tikzpicture}
\caption{Identified Sets Based on Cyclic Montononicity For the Trinary Choice Model Where Each Regressor Has a $s$-Point Support.}
\end{figure}

\subsection{Numerical Illustration}\label{sec:num}
In this subsection, we use a numerical example to illustrate the identifying power of cyclic monotonicity.  
We consider a three-choice model, where
the  ${X}_{it}^k$, is a  3-dimensional vector for $k=1,2$. Consider a two-period panel, i.e., $T=2$. Let $\{u_{it}^k\}_{k=0, 1, 2; ~t=1,2}$ be  independent type-I extreme value random variables, and let $\epsilon_{it}^k = u_{it}^k  - u_{it}^0$ for $k = 1,2; ~t=1,2$. Normalize $\beta_1 = 1$, and let the true value of $\beta_2, \beta_3$ both be $1$.  Let the support of $(X_{j, it}^k)_{k=1,2; ~j=1,2,3;~ t=1,2}$ be $\{1, 1/2, \dots, 1/s\}^{12}$. Let $A_i^1 = \omega_i^1 X_{1,i1}^1$ and $A_i^2 = \omega_i^2 x_{3,i1}^1$ for binary random variable $\omega_i^1$ and $\omega^2_i$. The variable $\omega_i^1$ takes the values  $1$ and $2$ each with probability $0.5$ and the variable $\omega^2_i$ takes the values $0$ and $-1$ each with probability $0.5$. The random variables $\omega_i^1$ and $\omega^2_i$ are mutually independent and are joint independent of  $(\vec{\epsilon}_{i1},\vec{\epsilon}_{i2})$. 

The identified sets based on length-2 cycles are drawn in Figure \ref{fig:tri}. We vary $s$ from $2$ to $4$ to see the change of $B_0$ with $s$.  As we can see, the identified set shrinks quickly as we add more support points.

\section{Estimation and Consistency}
Since the identification in this paper is based on inequalities rather than equalities, standard estimation and inference methods do not apply. Nevertheless, we propose a computationally easy consistent estimator for $\beta$. Confidence intervals for $\beta$ can be constructed using the methods proposed for conditional moment inequalities because the identifying conditions in (\ref{panelidm}) are conditional moment inequalities.\footnote{See, for example, Andrews and Shi (2013)  and Chernozhukov, Lee and Rosen (2013). These methods are partial-identification robust, and thus can be applied even when our point identification assumptions do not hold.}  Therefore, we do not discuss  it here.

The identification results presented above are based on length-2 cycles. Thus, we focus on the length-2 cycles for estimation as well, although in principle one could use cycles of any length. We do so partly for notational tractability, and partly because length 2-cycles are available for all panel data sets and are computationally simpler.


In the asymptotic analysis, we consider the case of a short panel; that is, the number of time period $T$ is fixed and the number of agents $n\to\infty$. Based on the panel data set, suppose that there is a uniformly consistent estimator $\hat{\vec{p}}_{j|s,t}(\vec{x}_s,\vec{x}_t)$ for $E(\vec{Y}_{ij}|\vec{X}_{it}=\vec{x}_t, \vec{X}_{is}=\vec{x}_s)$ for all $1\leq s<t\leq T$ and for $j=s,t$. Then a consistent estimator of $\beta$ can be obtained as $\widehat{\beta} = \widetilde{\beta}/\|\widetilde{\beta}\|$, where
\begin{equation}
\widetilde{\beta} =
\arg\min_{b\in R^{d_x}:\max_{j=1,\dots,J}|b_j|=1}Q_n(b),\quad\text{ and}
\end{equation} 
\begin{equation}
Q_n(b)=
\max_{1\leq s<t\leq T}n^{-1}\sum_{i=1}^n \left[(b'\vec{X}_{is} - b'\vec{X}_{it})(\hat{\vec{p}}_{s|s,t}(\vec{X}_{is}, \vec{X}_{it}) - \hat{\vec{p}}_{t|s,t}(\vec{X}_{is}, \vec{X}_{it}))\right]_{-},
\end{equation}
where $[x]_{-} = |\min\{x, 0\}|$. The estimator is easy to compute because $Q_n(b)$ is a convex function and the constraint set of the minimization problem is the union of $2d_x$ convex sets. \footnote{An alternative candidate for $\widehat{\beta}$ is $\arg\min_{b\in R^{d_x}:\|b\|=1} Q_n(b)$. However, obtaining this estimator requires minimizing a convex function on a non-convex set, which is computationally less attractive.}

The following theorem shows the consistency of $\widehat{\beta}$.
\begin{assumption}\label{ass:consistency}
\textup{(a)} The set $cc({\cal G}_M) $ is a half-space for $M=2$.

\textup{(b)} $\sup_{\vec{x}_{s},\vec{x}_t\in\vec{\cal X}}\sup_{j=s,t}\|\hat{\vec{p}}_{j|s,t}(\vec{x}_s,\vec{x}_t) - \vec{p}_{j|s,t}(\vec{x}_s,\vec{x}_t)\|\to_p 0$  as $n\to\infty$ for all $1\leq s<t\leq T$, where $\vec{p}_{j|s,t}(\vec{x}_s,\vec{x}_t) = E[\vec{Y}_{ij}|\vec{X}_{is}= \vec{x}_s, \vec{X}_{it} = \vec{x}_t]$ for $j=s,t$, and

\textup{(c)} $E[\|\vec{X}_{it}\|]\leq \infty$.

\end{assumption}
\begin{theorem}[Consistency]\label{thm:Consistency} Suppose that Assumptions \textup{\ref{ass:error}} and \textup{\ref{ass:consistency}} hold. Then, $\widehat{\beta}\to_p \beta$ as $n\to\infty$.
\end{theorem}

The consistency result in Theorem \ref{thm:Consistency} relies on a uniformly consistent estimator of the conditional choice probability $\vec{p}_{j|s,t}(\vec{x}_s,\vec{x}_t)$. Such estimators are abundant in the nonparametric regression literature; see for example, Cheng (1984)  for the $k$-nearest neighbor estimator, Chapter 2 of Li and Racine (2006) for kernel regression estimators, and Hirano, Imbens, and Ridder (2009) for a sieve logit estimator. Deriving the convergence rate of $\hat{\beta}$ appears to be a difficult problem and is left for future work.

\subsection{Monte Carlo Simulation}
Consider a trinary choice example and a two-period panel. Let $X_{it}^k$ be a three-dimensional covariate vector: $X_{it}^k = (X_{j,it}^k)_{j=1,2,3}$. Let $(X_{j,it}^k)_{j=1,2,3; k=1,2; t=1,2}$ be independent uniform random variables in $[0,1]$. Let $A_i^k = (\omega_i^k+\sum_{j=1}^3X_{j,i1}^k)/4$ for $k=1,2$, $t=1,2$, where $\omega_{i}^{k}$ is uniform in $[0,1]$, independent across $k$ and independent of other model primitives.
Let \begin{equation}
(u_{it}^0, u_{it}^1, u_{it}^2)\sim N\left(\begin{pmatrix}0\\0\\0\end{pmatrix},\begin{pmatrix}1&0&0\\0&1&0.5\\0&0.5&1\end{pmatrix}\right),
\end{equation}
 and let $\epsilon_{it}^k = A_i^k(u_{it}^k - u_{it}^0)$, for $t=1,2$. Let
 $({u}_{i1}^0, u_{i1}^1, u_{i1}^2)$ be independent of $({u}_{i2}^0,
 u_{i2}^1, u_{i2}^2)$. Let the true coefficient parameter $\beta =
 (1,0.5,0)$.  Note that for this test model, only our estimator yields
 consistent point estimates: Pakes and Porter (2013) only consider partial identification,
 and Chamberlain's (1980) conditional logit model requires the errors to be i.i.d. extreme-value distributed.
 
 We compute the bias, standard deviation (SD) and the root mean-squared error (rMSE) of each element of $\widehat{\beta}$ defined in the previous section. The nonparametric conditional choice probabilities are estimated using the $k$-nearest neighbor estimator where the tuning parameter $k$ is selected via leave-one-out cross-validation. We consider four sample sizes $250$, $500$, $1000$, and $2000$, and use $6000$ Monte Carlo repetitions.  The results are reported in Table \ref{tab:mc}. As we can see, the standard deviation decreases with the sample size for every element of the parameter, which is the general pattern for the bias as well. 

 \begin{table}[h!]
\begin{center}
\caption{Monte Carlo Results (6000 Repetitions)}\label{tab:mc}
\begin{tabular}{c|ccc|ccc|ccc}
\hline 
\multirow{2}{*}{$n$} & \multicolumn{3}{c|}{$\widehat{\beta}_{1}$} & \multicolumn{3}{c|}{$\widehat{\beta}_{2}$} & \multicolumn{3}{c}{$\widehat{\beta}_{3}$}\tabularnewline
\cline{2-10} 
 & BIAS & SD & rMSE & BIAS & SD & rMSE & BIAS & SD & rMSE\tabularnewline
\hline 
250 & .0158 & .0694 & .0712 & - .0890& .1375 & .1638 & -.0173 & .1384 & .1394\tabularnewline
\hline 
500 & .0199 & .0396 & .0444 & -.0682 & .0918 & .1143 & -.0128 & .1009 & .1017\tabularnewline
\hline 
1000 & .0183& .0288 & .0341 & -.0525 & .0664 &.0846& -.0149 & .0750 & .0764\tabularnewline
\hline 
2000 & .0163 & .0216 & .0270 & -.0412 & .0483 & .0635 & -.0147 & .0527 & .0547\tabularnewline
\hline 
\end{tabular}
\end{center}
\end{table}

\section{Related model: Aggregate Panel Multinomial Choice Model}
Up to this point, we have focused on the setting when the researcher has individual-level panel data on multinomial choice.
In this section, we discuss an important and simpler related model: the panel multinomial choice model estimated using {\em aggregate} data.   Such models are often encountered in empirical industrial organization.\footnote{See, for instance, Berry, Levinsohn, and Pakes (1995) and Berry and Haile (2014).}   In this setting, the researcher observes the aggregated choice probabilities (or {\em market shares}) for the consumer population in a number of regions and across a number of time periods. Correspondingly, the covariates are also only observed at region/time level for each choice option. To be precise, we observe $(\vec{S}_{ct}, \vec{X}_{ct}= (X^{1,'}_{ct},\dots,X^{K,'}_{ct})')_{c=1}^C{}_{t=1}^T$ which denote, respectively, the region/time-level choice probabilities and covariates.
Only a ``short'' panel is required, as our approach works with as few as two periods.

We model the individual choice $\vec{Y}_{ict} =(Y_{ict}^1, \dots,Y_{ict}^K)'$ as
\begin{equation}
Y_{ict}^k = 1\{\beta'X_{ct}^k + A_{c}^k + \epsilon_{ict}^k\geq \beta'X_{ct}^{k'} + A_{c}^{k'} + \epsilon_{ict}^{k'}~ \forall k'=0,\dots,K\},
\end{equation}
where $X_{ct}^0$, $A^0_c$, and $\epsilon_{ict}^0$ are normalized to zero, $\vec{A}_c = (A_c^0,\dots,A_c^K)'$ is the choice-specific regional fixed effect, and $\vec{\epsilon}_{ict} = (\epsilon_{ict}^1,\dots,\epsilon_{ict}^K)'$ is the vector of idiosyncratic shocks. 
(This is the main distinction vis-a-vis the individual-level model discussed previously, where the $A^k_i$ are choice- and {\em individual}-specific fixed effects.)
Correspondingly, the vector of choice probabilities
$\vec{S}_{ct}=(S^1_{ct},\dots,S^K_{ct})'$ is obtained as the fraction of $n_{ct}$ agents in region $c$ and time $t$ who chose option $k$, i.e.
$\vec{S}_{ct} = n_{ct}^{-1}\sum_{i=1}^{n_{ct}}\vec{Y}_{ict}$.
We make the following assumptions (which are weaker than the corresponding assumptions for individual-level panel data)

\begin{assumption}\label{ass:error-agg}
\textup{(a)} The error term $\vec{\epsilon}_{ict}$ is independent of $\vec{X}_{ct}$ given $\vec{A}_c$, and


\textup{(b)} the conditional c.d.f. of $\vec{\epsilon}_{ict}$ given  $\vec{A}_c$ is continuous everywhere.
\end{assumption}

Under Assumption \ref{ass:error-agg}, Lemma \ref{lem:gradient} implies that, for any cycle $t_1,t_2,\dots,t_M, t_{M+1}=t_1$ in  $\{1,\dots, T\}$, 
\begin{equation}
\sum_{m=1}^M E(\vec{Y}_{ict_m}^\prime|\vec{X}_{ct_m},\vec{A}_c)(\vec{X}_{c,t_{m}}^\prime\beta-\vec{X}^\prime_{c,t_{m+1}}\beta) \geq 0, ~a.s. \label{aPanelIneq2}
\end{equation}
Unlike in the the individual-level setting, we no longer need to eliminate the fixed effect $\vec{A}_c$ from the inequalities in Eq. (\ref{aPanelIneq2}). This is because in the aggregate setting, even though $\vec{A}_c$ is latent, the conditional choice probability $E(\vec{Y}_{ict}|\vec{X}_{ct},\vec{A}_c)$ can be estimated uniform consistently by $\vec{S}_{ct}$.\footnote{Specifically, we can use $\vec{S}_{ct} = n_{ct}^{-1}\sum_{i=1}^{n_{ct}}\vec{Y}_{ict}$ to estimate $E(\vec{Y}_{ict}|\vec{X}_{ct}, \vec{A}_{c})$. If $\inf_{c,t}n_{ct}$ grows fast enough with $C\times T$, this estimator is uniformly consistent, i.e. $\sup_c\sup_t\|\vec{S}_{ct} - E(\vec{Y}_{ict}|\vec{X}_{ct}, \vec{A}_{c})\|\to_p 0$. 
Section 3.2 of Freyberger's (2013) arguments (using Bernstein's Inequality) imply that the above convergence holds if $\log(C\times T)/\min_{c,t}n_{ct}\to 0$.  
} Because no integrating-out is needed, we also do not need the conditional irrelevance assumption analogous to Assumption \ref{ass:error}(b), which enables us to allow $\vec{X}_{ct}$ to contain lagged values of $\vec{S}_{ct}$.

Using $\vec{S}_{ct}$ as the estimator of $E[\vec{Y}_{ict}|\vec{X}_{ct},\vec{A}_c]$, we can construct a consistent estimator of $\beta$: $\widehat{\beta} = \widetilde{\beta}/\|\widetilde{\beta}\|$, where
\begin{equation}
\widetilde{\beta} =
\arg\min_{b\in R^{d_x}:\max_{j=1,\dots,J}|b_j|=1}Q_n(b),\quad\text{ and}
\end{equation} 
\begin{equation}
Q_n(b)=
\max_{1\leq s<t\leq T}C^{-1}\sum_{c=1}^C \left[(b'\vec{X}_{cs} - b'\vec{X}_{ct})(\vec{S}_{cs}- \vec{S}_{ct})\right]_{-}.
\end{equation}
This estimator is consistent by similar arguments as those for Theorem \ref{thm:Consistency}.

\subsection{Empirical Illustration}
Here we consider an empirical illustration, based on the aggregate panel multinomial choice model described above.  We estimate a discrete
choice demand model for bathroom tissue, using store/week-level scanner
data from different branches of Dominicks supermarket.\footnote{This
  dataset has previously been used in many papers in both economics
  and marketing; see a partial list at {\sffamily http://research.chicagobooth.edu/kilts/marketing-databases/dominicks/papers}.}   The bathroom tissue category is convenient because there are relatively few brands of bathroom tissue, which simplifies the analysis.   The data are collected at the store and week level, and report sales and prices of different brands of bathroom tissue.   For each of 54 Dominicks stores, we aggregate the store-level sales of bathroom tissue up to the largest six brands, lumping the remaining brands into the seventh good (see Table \ref{products}).  
\begin{table}[h!]
   \centering
    \caption{\footnotesize Table of the 7 product-aggregates    used in estimation.              }\label{products}
   \linespread{1}\small
  \begin{tabular}{|c|l|}\hline
&Products included in analysis \\\hline
1 &Charmin\\
2 &White Cloud\\
3 &Dominicks\\
4 &Northern\\
5 &Scott\\
6 &Cottonelle\\
7 &Other good (incl. Angelsoft, Kleenex, Coronet and smaller brands)\\\hline
\end{tabular}
\end{table}

We form moment conditions based on cycles over weeks, for each store.   In the estimation results below, we consider cycles of length 2. 
 Since data are observed at the weekly level, we consider subsamples of 10 weeks or 15 weeks which were drawn at periodic intervals from the 1989-1993 sample period. After the specific weeks are drawn, all length-2 cycles that can be formed from those weeks are used. 

We allow for store/brand level fixed effects and use the techniques developed in Section 3.1 to difference them out. Due to this, any time-invariant brand- or store-level variables will be subsumed into the fixed effect, leaving only explanatory covariates which vary both across stores and time.   As such, we consider a simple specification with $X^k = (\text{PRICE, DEAL, PRICE*DEAL})$.   PRICE is measured in dollars per roll of bathroom tissue, while DEAL is defined as whether a given brand was on sale in a given store-week.\footnote{The variable DEAL takes the binary values $\{0,1\}$ for products 1-6, but takes continuous values between 0 and 1 for product 7. The continuous values for product 7 stand for the average on-sale frequency of all the small brands included in the product-aggregate 7. This and the fact that PRICE is a continuous variable make the point identification condition, Assumption \ref{ass:sufbdm}, plausible.}  Since any price discounts during a sale will be captured in the PRICE variable itself, DEAL captures any additional effects that a sale has on behavior, beyond price.
Summary statistics for these variables are reported in  Table \ref{tab:summary}.

\begin{table}[h!]
\begin{center}
\caption{Summary Statistics}\label{tab:summary}
\begin{tabular}{ccccccc}
& &min& max&mean &median &std.dev \\
\hline
\hline
10 week data&DEAL&0&1&.4350&0&.4749\\
&PRICE&.1776&.6200&.3637&.3541&.0876\\
\hline
15 week data&DEAL&0&1&.4488&0&.4845\\
&PRICE&.1849&.6200&.3650&.3532&.0887\\
\hline\hline
\end{tabular}
\end{center}
\end{table}

The point estimates are reported in Table \ref{demandpointtable}.   One interesting observation from the table is that the sign of the interaction term is negative, indicating that consumers 
are more price sensitive when a product is on sale. This may be consistent with the story that the sale status draws consumers' attention to price (from other characteristics of the product). 

\begin{table}[h!]
\begin{center}
\caption{Point Estimates for Bathroom Tissue Choice Model}\label{demandpointtable}
\begin{tabular}{ccccc}
&& 10 week data&15 week data  \\
\hline\hline
$\beta_{1}$ & deal &          .1053     &.0725\\
 \hline 
$\beta_{2}$ & price &         -.9720 &-.9922 \\
 \hline
$\beta_{3}$ & price*deal   & -.2099    &-.1017   \\

\hline\hline
\end{tabular}

\end{center}
\end{table}

\section{Conclusions}
In this paper we explored how the notion of cyclic monotonicity can be exploited for the identification
and estimation of panel multinomial choice models with fixed effects.   In these models, the
social surplus (expected maximum utility) function  is convex, implying that its gradient, which
corresponds to the choice probabilities, satisfies cyclic
monotonicity.   This is just the appropriate generalization of the
fact that the slope of a single-variate convex function is
non-decreasing. In ongoing work, we are considering the possible extension of these
ideas to other models and economic settings.   

Throughout this paper, we have focused on estimation under the assumption that the conditions for point identification are satisfied.   In the case that these conditions are not satisfied, the parameters will only be partially identfied, and we might consider an alternative inferential approach for this case based on recent work by Freyberger and Horowitz (2013).   Since this approach is quite different in spirit to the methods described so far, we do not discuss it here.

\newpage

\appendix
\section{Appendix: Proofs}\label{proofs}

\begin{proof}[Proof of Lemma \textup{\ref{lem:gradient}}] 

(a)  By the independence between $\vec{U}$ and $\vec{\epsilon}$, we have
\begin{equation}
{\cal G}(\vec{u}) = E\{\max_{k}[U^k+\epsilon^k]|\vec{U}=\vec{u}\} = E\{\max_k[u^k+\epsilon^k]\}.
\end{equation}
This function is convex because $\max_k[u^k+\epsilon^k]$ is convex for all values of $\epsilon^k$ and the expectation operator is linear. 

(b,c) Without loss of generality, we focus on the differentiability with respect to $u^K$. Let $(u^1_\ast,\dots,u^K_\ast)$ denote an arbitrary fixed value of $(U^1,\dots, U^K)$, and let $u_\ast^0 =0$. It suffices to show that $\lim_{\eta\to 0} [{\cal G}(u_\ast^1, \allowbreak\dots, u_\ast^{K}+\eta) - {\cal G}(u_\ast^1,\dots, u_\ast^{K})]/\eta$ exists.  We show this using the bounded convergence theorem. First observe that
\begin{align}
\frac{{\cal G}(u_\ast^1,\dots,u_\ast^{K}+\eta) - {\cal G}(u_\ast^1,\dots,u_\ast^K)}{\eta} = E\left[\frac{\Delta(\eta,\vec{u}_\ast,\vec{\epsilon})}{\eta} \right],
\end{align}
where $\Delta(\eta,\vec{u}_\ast,\vec{\epsilon}) = \max\{u_\ast^1+\epsilon^1,\dots,u_\ast^{K}+\eta+\epsilon^K\} -\max\{u_\ast^1+\epsilon^1,\dots,u_\ast^{K}+\epsilon^K\}$. Consider an arbitrary value $\vec{e}$ of $\vec{\epsilon}$ and $e^0=0$.  If $e^K+u_\ast^K>\max_{k=0,\dots,K-1}[u_\ast^k+e^k]$, for $\eta$ close enough to zero, we have
\begin{equation}
\frac{\Delta(\eta,\vec{u}_\ast,\vec{e}) }{\eta}=\frac{ (u_\ast^K+\eta+e^K)-(u_\ast^K+e^K)}{\eta} = 1.
\end{equation}
Thus,
\begin{equation}
\lim_{\eta\to0} \frac{\Delta(\eta,\vec{u}_\ast,\vec{e})}{\eta}=1.
\end{equation}
On the other hand, if $e^K+u_\ast^K<\max_{k=0,\dots,K-1}[u_\ast^k+e^k]$, then for $\eta$ close enough to zero, we have
\begin{equation}
\frac{\Delta(\eta,\vec{u}_\ast,\vec{e})} {\eta}=\frac{ 0}{\eta} = 0.
\end{equation}
Thus,\begin{equation}
\lim_{\eta\to0} \frac{\Delta(\eta,\vec{u}_\ast,\vec{e})}{\eta}=0.
\end{equation}
Because $\vec{\epsilon}$ has a continuous distribution, we have $\Pr(\epsilon^K+u_\ast^K=\max_{k=0,\dots,K-1}[u_\ast^k+\epsilon^k]) = 0$. Therefore, almost surely,
\begin{equation}
\lim_{\eta\to0} \frac{\Delta(\eta,\vec{u}_\ast,\vec{\epsilon})}{\eta} = 1\{\epsilon^K+u_\ast^K>\max_{k=0,\dots,K-1}[u_\ast^k+\epsilon^k]\}.
\end{equation}

Also, observe that
\begin{equation}
\left|\frac{\Delta(\eta,\vec{u}_\ast, \vec{\epsilon})}{\eta}\right| \leq \left|\frac{u_\ast^K+\eta+\epsilon^K  - (u_\ast^K+\epsilon^K)}{\eta} \right|=1<\infty.
\end{equation}
Thus, the bounded convergence theorem applies and yields
\begin{equation}
\lim_{\eta\to 0}E\left[\frac{\Delta(\eta,\vec{u}_\ast,\vec{\epsilon})}{\eta} \right] = E[1\{\epsilon^K+u_\ast^K>\max_{k=0,\dots,K-1}[u_\ast^k+\epsilon^k]\}] = p^K(\vec{u}).
\end{equation}
This shows both part (b) and part (c).

Part (d) is a direct consequence of part (c) and Proposition \ref{prop:Convexity}.
\end{proof}

\bigskip

\begin{proof}[Proof of Theorem {\em\ref{thm:setid}}] 
First note that the true parameter value $\beta$ satisfies $\beta'g\geq 0$ for all $g\in {\cal G}$ by (\ref{panelid2}) and by the definition of ${\cal G}$. Then by the definition of $cc({\cal G})$, we have $\beta'g\geq 0$ for all $g\in cc({\cal G})$. That implies that $cc({\cal G})$ is a subset of the half-space $\{g\in R^{d_x}: \beta'g\geq 0\}$. By the same logic, we have
\begin{equation}
cc({\cal G})\subseteq \cap_{b\in B_0}\{g\in R^{d_x}: b'g\geq 0\}.
\end{equation}
Note that $0$ is on the boundary of all the half-spaces $\{g\in R^{d_x}: b'g\geq 0\}$, which implies that   the intersection of these half-spaces is a half-space if and only if they are all the same half-space, or equivalently if and only if $B_0$ is a singleton. This proves the theorem. 
\end{proof}
\bigskip

\begin{proof}[Proof of Theorem \textup{\ref{thm:sufbdm}}] Consider the set 
\begin{equation}
\tilde{G} = \{g\in G: \beta'g>0\}.
\end{equation}
Next, we show that (i) $cone({\tilde G})$ is dense in the set $\{g\in R^{d_x}:\beta'g> 0\}$ and (ii) $cone(\tilde{G}) \subseteq cone({\cal G})$. Both (i) and (ii) together immediately implies that $ \{g\in R^{d_x}:\beta'g\geq 0\}\subseteq closure(cone({\cal G}))\subseteq cc({\cal G})$. By (\ref{panelidm}), $cc({\cal G})\subseteq (\{g\in R^{d_x}:\beta'g\geq 0\})$. Thus  $cc({\cal G}) = \{g\in R^{d_x}:\beta'g\geq 0\}$.

To show (i), consider an arbitrary point $g_0\in R^{d_x}$ such that $\beta'g_0>0$. Then by Assumption \ref{ass:sufbdm}, there exists sequences $\{\lambda_m\in [0,\infty)\}$ and $\{g_m\in G\}$ such that $\lim_{m\to\infty} \lambda_m g_m = g_0$. Because $\beta'g_0>0$, there must exists an $M>0$ such that for all $m>M$, $\lambda_m\beta'g_m>0$. For these $m$'s, we must have $\beta'g_m>0$. That is, $g_m\in \tilde{G}$. Therefore, $g_0$ can be approximately arbitrarily closely by points in $cone(\tilde{G})$, which shows result (i). 

To show (ii), consider an arbitrary point $\tilde{g}\in cone({\tilde G})$. Then there exists $\lambda\geq 0$ and $g\in {G}$ such that $\beta'g>0$ and $\tilde{g} = \lambda g$. By the definition of $G$, there exist $s, t \in\{1,\dots,T\}$, $k \in\{ 1,\dots,K\}$  and $\vec{x}^{-k} \in \vec{\cal X}_{-k}$ such that $g\in G_{s,t}^{k}(\vec{x}^{-k})$. Then there exists $x^{k}_\ast$ and $x^{k}_{\dagger}$ such that $x^{k}_{\ast} - x^{k}_{\dagger}= g$ and $(x^{k}_{\ast}, x^{k}_{\dagger})$ is in the conditional support of $(X^k_{it},X^k_{is})$ given $\vec{X}^{-k}_{it} = \vec{X}_{is}^{-k} = \vec{x}^{-k}$. By the definition of ${\cal G}$, the following element belongs to ${\cal G}$:
\begin{equation}
E[Y_{it}^k - Y_{is}^k|{X}_{it}^k = x_\ast^k, \vec{X}_{it}^{-k} =\vec{X}_{is}^{-k} = \vec{x}^{-k}, X_{is}^k = x_\dagger^k]g
\end{equation}
Below we show that
\begin{align}
E[Y_{it}^k - Y_{is}^k|{X}_{it}^k = x_\ast^k, \vec{X}_{it}^{-k} =\vec{X}_{is}^{-k} = \vec{x}^{-k}, X_{is}^k = x_\dagger^k]>0. \label{prob}
\end{align}
This implies that $g\in cone({\cal G})$. Thus, $\tilde{g}\in cone({\cal G})$, which shows result (ii). 

The result in (\ref{prob}) follows from the derivation:
\begin{align}
&E[Y_{it}^k|{X}_{it}^k = x_\ast^k, \vec{X}_{it}^{-k} =\vec{X}_{is}^{-k} = \vec{x}^{-k}, X_{is}^k = x_\dagger^k]\nonumber\\
& = E\{E[Y_{it}^k|{X}_{it}^k = x_\ast^k, \vec{X}_{it}^{-k} =\vec{X}_{is}^{-k} =\vec{x}^{-k}, \vec{A}_i]\}\nonumber\\
& = E\{E[Y_{it}^k|{X}_{it}^k = x_\ast^k, \vec{X}_{it}^{-k} = \vec{x}^{-k}, \vec{A}_i]\}\nonumber\\
& =E\left[ \Pr\left(\beta'x_\ast^k+A_i^k+\epsilon^k_{it}\geq \max_{k' = 0, 1,\dots, k-1, k+1,\dots,K}\beta'x^{k'}+A_i^{k'}+\epsilon^{k'}_{it}|{X}_{it}^k = x_\ast^k, \vec{X}_{it}^{-k} = \vec{x}^{-k}, \vec{A}_i\right)\right]\nonumber\\
&=E\left[ \Pr\left(\beta'x_\ast^k+A_i^k+\epsilon^k_{it}\geq \max_{k' = 0, 1,\dots, k-1, k+1,\dots,K}\beta'x^{k'}+A_i^{k'}+\epsilon^{k'}_{it}| \vec{A}_i\right)\right]\nonumber\\
&>E\left[ \Pr\left(\beta'x_\dagger^k+A_i^k+\epsilon^k_{it}\geq \max_{k' = 0, 1,\dots, k-1, k+1,\dots,K}\beta'x^{k'}+A_i^{k'}+\epsilon^{k'}_{it}| \vec{A}_i\right)\right]\nonumber\\
&=E\left[ \Pr\left(\beta'x_\dagger^k+A_i^k+\epsilon^k_{is}\geq \max_{k' = 0, 1,\dots, k-1, k+1,\dots,K}\beta'x^{k'}+A_i^{k'}+\epsilon^{k'}_{is}| \vec{A}_i\right)\right]\nonumber\\
&=E[Y_{is}^k|{X}_{it}^k = x_\ast^k, \vec{X}_{it}^{-k} =\vec{X}_{is}^{-k} = \vec{x}^{-k}, X_{is}^k = x_\dagger^k],
\end{align}
where the first equality holds by the law of iterated expectations, the second equality holds by Assumption \ref{ass:error}(b), the third equality holds by the specification of the multinomial choice model, the fourth equality holds by Assumption \ref{ass:error}(a), the inequality holds by Assumption \ref{ass:error2m} and $\beta'(x_\ast^k - x_\dagger^k)>0$, the fifth equality holds by stationarity, and the last equality follows by analogous arguments as those preceding the inequality.\end{proof}

\bigskip
\begin{proof}[Proof of Theorem \textup{\ref{thm:suflgm}}] Consider the set
\begin{align}
\tilde{G} = \{g\in G: \beta'g>0\}.
\end{align}
It has been shown in the proof of Theorem \ref{thm:sufbdm} that $cone(\tilde{G})\subseteq cone({\cal G})$ under Assumptions \ref{ass:error}(a)-(b) and \ref{ass:error2m}. That implies $cc(\tilde{G})\subseteq cc({\cal G})$. Below we show that $\{g\in R^{d_x}:\beta'g\geq 0\}\subseteq cc(\tilde{G})$. This implies $\{g\in R^{d_x}:\beta'g\geq 0\}\subseteq cc({\cal G})$. By the definition of ${\cal G}$ and by (\ref{panelidm}), $cc({\cal G})\subseteq \{g\in R^{d_x}:\beta'g\geq 0\}$. Therefore, $cc({\cal G})= \{g\in R^{d_x}:\beta'g\geq 0\}$, which proves the theorem.

Now we show $\{g\in R^{d_x}:\beta'g\geq 0\}\subseteq cc(\tilde{G})$. Suppose that $\beta_j>(<)0$. First, by Assumption \ref{ass:suflgm}(a), we have that
\begin{align}
{\tilde G} = \{g\in R^{d_x}:  g_{-j}\in G_{-j}, g_j>(<) - \beta_{-j}'g_{-j}/\beta_j\},\label{Gtilde}
\end{align}
where $\beta_{-j} = (\beta_1,\dots,\beta_{j-1},\beta_{j+1},\dots,\beta_{d_x})'$. Consider an arbitrary point $g_0\in \{g\in R^{d_x}:\beta'g\geq 0\}$. Then, $g_{0,j}>(<)-g_{0,-j}'\beta_{-j}/\beta_j$. Let  \begin{equation}
d= g_{0,j}+g_{0,-j}'\beta_{-j}/\beta_j.
\end{equation}
Then, $d>(<) 0$. 

By Assumption \ref{ass:suflgm}(b), $G_{-j}$ spans $R^{d_x-1}$. By definition, $G$ is symmetric about the origin, which implies that $G_{-j}$ is also symmetric about the origin. Thus, $G_{-j}$ spans $R^{d_x-1}$ with nonnegative weights. Thus, there exists an integer $M$, weights $c_1, \dots, c_M>0$, and $g_{1,-j} ,\dots, g_{M,-j}\in G_{-j}$ such that $g_{0,-j} = \sum_{m=1}^Mc_m g_{m,-j}$. 
Let $g_{m,j} = \left(d/\sum_{m=1}^Mc_m\right)-\left(g_{m,-j}'\beta_{-j}/\beta_j\right)$ for $m=1,\dots,M$. Let $g_m$ be the vector whose $j$th element is $g_{m,j}$ and the rest of whose elements form $g_{m,-j}$, for $m=1,\dots, M$. Then $g_m\in {\tilde G}$ for $m=1,\dots,M$ (according to Eqn. (\ref{Gtilde})), and $g_0 = \sum_{m=1}^M c_m g_m$. Thus, $g_0\in cc(\tilde G)$. Therefore, $ \{g\in R^{d_x}: \beta'g\geq 0\}\subseteq cc({\tilde G})$.
\end{proof}

\bigskip

\begin{proof}[Proof of Theorem \textup{\ref{thm:Consistency}}] For any $b\in R^{d_x}$, let $\|b\|_\infty = \max_{j=1,\dots,J}|b_j|$. Below we show that \begin{equation}
\widetilde{\beta}\to_p \beta/\|\beta\|_\infty .\label{maxnorm}
\end{equation}
   This implies that $\widehat{\beta}\to_p \beta$ because $\widehat{\beta} = \widetilde{\beta}/\|\widetilde{\beta}\|$ and the mapping $f:\{b\in R^{d_x}:\|b\|_\infty=1\}\to \{b\in R^{d_x}:\|b\|=1\}$ such that $f(b) = b/\|b\|$ is continuous.

Now we show Eqn. (\ref{maxnorm}). Let
\begin{equation}
Q(b) = \max_{1\leq s<t\leq T}E\left[b'(\vec{X}_{is}-\vec{X}_{it})\left(\vec{p}_{s|s,t}(\vec{X}_{is},\vec{X}_{it}) -\vec{p}_{t|s,t}(\vec{X}_{is},\vec{X}_{it})\right) \right]_{-}.
\end{equation}
Under Assumption \ref{ass:error}, the identifying inequalities (\ref{panelidm}) hold, which implies that \begin{equation}
Q(\beta) = Q( \beta/\|\beta\|_\infty) )= 0.\end{equation}
 Assumption \ref{ass:consistency} and Theorem \ref{thm:setid} together imply that, for any $b\neq \beta/\|\beta\|_\infty$ such that $\|b\|_\infty=1$, 
\begin{align}
\Pr\left(b'(\vec{X}_{is}-\vec{X}_{it})\left(\vec{p}_{s|s,t}(\vec{X}_{is},\vec{X}_{it}) -\vec{p}_{t|s,t}(\vec{X}_{is},\vec{X}_{it})\right)<0\right)>0. 
\end{align}
Thus, for any $b\neq \beta/\|\beta\|_\infty$ such that $\|b\|_\infty=1$, we have that $Q(b)>0$. This, the continuity of $Q(b)$, and the compactness of the parameter space $\{b\in R^{d_x}:\|b\|_\infty=1\}$ together imply that, for any $\varepsilon>0$, there exists a $\delta>0$ such that,
\begin{equation}
\inf_{b\in R^{d_x}:\|b\|_\infty=1, \|b-\beta\|>\varepsilon}Q(b) \geq \delta. 
\end{equation}
If in addition, we can show the uniform convergence of $Q_n(b)$ to $Q(b)$, then the consistency of $\widehat{\beta}$ follows from standard consistency arguments (see, e.g., Newey and McFadden (1994)). 

Now we show the uniform convergence of $Q_n(b)$ to $Q(b)$. That is, we show that
\begin{equation}
\sup_{b\in R^{d_x}:\|b\|_\infty=1}|Q(b) - Q_n(b)|\to_p 0.\label{Uconv}
\end{equation}
First, we show the stochastic equicontinuity of $Q_n(b)$. For any $b,b^\ast\in R^{d_x}$ such that $\|b\|_\infty=\|b_j^\ast\|_\infty=1$, consider the following derivation:
\begin{align}
&|Q_n(b)-Q_n(b^\ast)|\nonumber\\& \leq \max_{1\leq s<t \leq T} n^{-1}\sum_{i=1}^n\left|(b-b^\ast)'(\vec{X}_{is} -\vec{ X}_{it})\left(\hat{\vec{p}}_{s|s,t}(\vec{X}_{is},\vec{X}_{it})-\hat{\vec{p}}_{t|s,t}(\vec{X}_{is},\vec{X}_{it})\right)\right|\nonumber\\
&\leq \max_{1\leq s<t \leq T} n^{-1}\sum_{i=1}^n\|b-b^\ast\|\|(\vec{X}_{is} -\vec{ X}_{it})\left(\hat{\vec{p}}_{s|s,t}(\vec{X}_{is},\vec{X}_{it})-\hat{\vec{p}}_{t|s,t}(\vec{X}_{is},\vec{X}_{it})\right)\|\nonumber\\
&\leq 2\max_{1\leq s<t \leq T} n^{-1}\sum_{i=1}^n\|\vec{X}_{is} -\vec{ X}_{it}\|\|b-b^\ast\|.\label{lip}
\end{align}
Therefore, for any fixed $\varepsilon>0$, we have
\begin{align}
&\lim_{\delta\downarrow 0}\underset{n\to\infty}{\lim\sup}\Pr\left(\sup_{b,b^\ast\in R^{d_x}, \|b\|_\infty=\|b^\ast\|_\infty=1,\|b-b^\ast\|\leq \delta}|Q_n(b)-Q_n(b^\ast)|>\varepsilon\right) \nonumber\\
&\leq\lim_{\delta\downarrow 0} \underset{n\to\infty}{\lim\sup}\Pr\left(2\delta\max_{1\leq s<t \leq T} n^{-1}\sum_{i=1}^n\|\vec{X}_{is} -\vec{ X}_{it}\|>\varepsilon\right)\nonumber\\
&\leq \lim_{\delta\downarrow 0}\underset{n\to\infty}{\lim\sup}\Pr\left(2\max_{1\leq s<t \leq T} n^{-1}\sum_{i=1}^n\|\vec{X}_{is} -\vec{ X}_{it}\|>\varepsilon/\delta\right)
\nonumber\\
&=0,
\end{align}
where the first inequality holds by (\ref{lip}) and the equality holds by Assumption \ref{ass:consistency}(c). This shows the stochastic equicontinuity of $Q_n(b)$. 

Given the stochastic equicontinuity $Q_n(b)$ and the compactness of $\{b\in R^{d_x}:\|b\|_\infty=1\}$, to show (\ref{Uconv}), it suffices to show that  for all $b\in R^{d_x}:\|b\|_\infty=1$, we have
\begin{equation}
Q_n(b)\to_p Q(b).
\end{equation}
For this purpose, let
\begin{equation}
\widetilde{Q}_n(b)=
\max_{1\leq s<t\leq T}n^{-1}\sum_{i=1}^n \left[(b'\vec{X}_{is} - b'\vec{X}_{it})({\vec{p}}_{s|s,t}(\vec{X}_{is}, \vec{X}_{it}) - {\vec{p}}_{t|s,t}(\vec{X}_{is}, \vec{X}_{it}))\right]_{-}.
\end{equation}
By Assumption \ref{ass:consistency}(c) and the law of large numbers, we have $\widetilde{Q}_n(b)\to_p Q(b)$. Now we only need to show that $|\widetilde{Q}_n(b)-Q_n(b)|\to_p 0$. But that follows from the derivation:
\begin{align}
&|\widetilde{Q}_n(b) - Q_n(b)|\nonumber\\
&\leq\max_{1\leq s<t\leq T}n^{-1}\sum_{i=1}^n \sum_{j=s,t}\left|(b'\vec{X}_{is} - b'\vec{X}_{it})(\hat{\vec{p}}_{j|s,t}(\vec{X}_{is}, \vec{X}_{it}) -{\vec{p}}_{j|s,t}(\vec{X}_{is}, \vec{X}_{it}) )\right|,\nonumber\\
&\leq2\sup_{\vec{x}_s,\vec{x}_t\in\vec{\cal X}}\sup_{j=s,t}\|\hat{\vec{p}}_{j|s,t}(\vec{x}_s,\vec{x}_t)- \vec{p}_{j|s,t}(\vec{x}_s,\vec{x}_t)\| \max_{1\leq s<t\leq T}n^{-1}\sum_{i=1}^n \|b'\vec{X}_{is} - b'\vec{X}_{it}\|,\nonumber\\
&\to_p 0,
\end{align}
where the convergence holds by Assumptions \ref{ass:consistency}(b)-(c). Therefore the theorem is proved. \end{proof}

\section{Appendix: Primitive necessary condition for point identification}\label{primnec}
In this section we characterize a primitive necessary condition for point identification, in the special case of a binary choice model.\footnote{We were not able to obtain an analogous result in the more general multinomial choice case because (i) cycles longer than 2 would need to be considered, and the simultaneous variation of $X_{it}^k$ for all $k$ would also need to be taken into account. }

In the binary choice case, $p(v,a)$ for each $a$ is a mapping from $R$ to $R$. For such mappings, the cyclic monotonicity is equivalent to monotonicity and it is without loss to consider only cycles of length 2.  Moreover, because $K=1$, there is no need for the $\vec{\cdot}$ sign on $X_{it}$, $\epsilon_{it}$, $A_i$, $v$, $a$, and $p(\cdot,\cdot)$. Similarly, there is also no need for the choice index superscript on these symbols. Thus, we omit them in this section. 

The set $G$ defined in (\ref{G}) simplifies to: 
\begin{equation}
G\equiv \cup_{s,t=1,\dots,T}G_{s,t},
\end{equation}
where $G_{s,t}$ is the support of $X_{it}-X_{is}$ for $s,t=1,\dots, T$. Theorem \ref{thm:nec} below is the main result of this section. It shows that, if one regressor has finite support and all other regressors have bounded support, then point identification cannot be achieved at all values of $\beta$.
\begin{assumption}\label{ass:nec} 
For some $j=1,\dots, d_x$,
\textup{(a)} $G_j$ is a finite set, and

 \textup{(b)} $G_{j'}$ is a bounded set for al $j'\neq j$. 
\end{assumption}
\begin{theorem}[{\bfseries Necessary conditions for point identification}]\label{thm:nec} Under Assumptions \textup{\ref{ass:error}(a)-(b)} and \textup{\ref{ass:error2m}}, if Assumption \textup{\ref{ass:nec}}  holds, then $cc({\cal G})$ is not always a half-space.
\end{theorem}

\noindent\textbf{Remark.} According to the Theorem \ref{thm:nec}, if one coordinate of $X_{is}-X_{it}$ has finite support for all $s,t$, then another coordinate of it must have unbounded support for some pair $(s,t)$. The variable $X_{j,is}-X_{j,it}$ may have finite support, either when $X_{j,it}$ has finite support, or when the change of $X_{j,it}$ across time periods is restricted to a few grids. When that is the case, point identification requires that another regressor, say, $X_{j',it}$ to changes unboundedly as $t$ changes. 

Theorem \ref{thm:nec} does not imply that $cc({\cal G})$ can never be a half-space. There can be $\beta$ values such that, when the population is generated from the model specified in (\ref{eq:panelspec}) and (\ref{eq:panelchoice}) with $\beta$ being that value,  $cc({\cal G})$ is a half-space. In other words, under the conditions of the theorem, point identification may be achieved in part of the parameter space, but not on the whole space of $\beta$. \hfill$\blacksquare$

\begin{proof}[Proof of Theorem \textup{\ref{thm:nec}}] It suffices to find at least one $\beta$ value that generates a population for which $cc({\cal G})$ is not a half-space. Below we find such a value among $\beta$'s that satisfy $\beta_j>0$, $\beta_{j^\ast}>0$ for some $j^\ast\neq j$, and $\beta_{j'} = 0$ for $j'\neq j, j^\ast$. It is useful to note that $G$ is symmetric about the origin by definition. So are $G_{j'}$'s for all $j'=1,\dots, d_x$. 

We discuss two cases. In the first case, $G_{j}\cap (-\infty, 0) = \emptyset$.  Then $G_{j} = \{0\}$ because it is symmetric about the origin. Then ${\cal G}$ is contained in the subspace $\{g\in R^{d_x}:g_{j}=0\}$. By the definition of $cc(\cdot)$, $cc({\cal G})$ must also be contained in $\{g\in R^{d_x}:g_{j}=0\}$, and thus cannot be a half-space of $R^{d_x}$. 

In the second case,  $G_{j}\cap (-\infty, 0)\neq \emptyset$. Assumption \ref{ass:nec}(a) implies that it is a finite set. Then $\eta \equiv \max (G_{j}\cap (-\infty, 0))$ is well defined and $\eta<0$. Assumption \ref{ass:nec}(b) implies that there is a positive constant $C$ such that $G_{j^\ast}\subseteq [-C,C]$. Let $\beta$ further satisfy $\beta_{j^\ast}/\beta_{j} < -\eta/C$. Then, for all $g\in G$ such that $g_{j}<0$, we have
\begin{align}
\beta'g = \beta_jg_j + \beta_{j^\ast}g_{j^\ast}\leq \beta_j\eta+\beta_{j^\ast}C<0. \label{betag}
\end{align}
Consider  $\tilde{G} = \{g\in G:\beta'g>0\}$. Then (\ref{betag}) implies that for all  $g\in G$ such that $g_{j}<0$, we have $g\notin\tilde{G}$. That implies that $cc(\tilde{ G})$ contains no point whose $j$th element is negative. Thus $cc(\tilde{G})$ is a proper subspace of $\{g\in R^{d_x}:\beta'g\geq 0\}$. The proof of Theorem \ref{thm:sufbdm} shows that $cone(\tilde{ G}) = cone({\cal G})$ under Assumptions \ref{ass:error}(a)-(b) and \ref{ass:error2m}, which implies that $cc(\tilde{G}) = cc({\cal G})$. Thus, $cc({\cal G})$ is a also a proper subset of $\{g\in R^{d_x}:\beta'g\geq 0\}$ and cannot be a half-space.
\end{proof}
\medskip

\section{Appendix: Relaxing the Independence Assumption Using Control Variables}\label{sec:control}
In this section, we use control variables to relax the conditional independence assumption -- Assumption \ref{ass:error}(a). Similar uses of control variables are common in regression analysis; see, for example, Chapter 7.5 of Stock and Watson (2010), and in the treatment effect
literature; see for example Rosenbaum and Rubin (1983) and Imbens (2004). 

We discuss two cases depending on whether the controls are time variant.

\subsection{Time-Invariant Controls.} 
Let $\eta_i$ be a vector of individual characteristics. Suppose that instead of Assumption \ref{ass:error}, we have Assumption \ref{ass:errorcontrol1}:
\begin{assumption}\label{ass:errorcontrol1} \textup{(a)} The error term $\vec{\epsilon}_{it}$ is independent of $\vec{X}_{it}$ given $\vec{A}_i$ and $\eta_i$,

\textup{(b)} $E[\vec{Y}_{it}|\vec{X}_{i1},\dots,\vec{X}_{iT},\vec{A}_i,\eta_i] = E[\vec{Y}_{it}|\vec{X}_{it},\vec{A}_i,\eta_i]$, for all $t$, and

\textup{(c)} the conditional c.d.f. of $\vec{\epsilon}_{it} $ given $\vec{A}_i$ and $\eta_i$ is continuous everywhere.
\end{assumption}
In addition, suppose that instead of Assumption \ref{ass:error2m}, we have Assumption \ref{ass:error2control1}:
\begin{assumption}\label{ass:error2control1} For every $k=1,\dots,K$, with positive probability, the conditional support of $\vec{\epsilon}_{it}$ given $\vec{A}_i$ and $\eta_i$ is $R^K$. 
\end{assumption}
Then all the results in the previous sections hold because the proofs of those results go through with Assumptions \ref{ass:error} and \ref{ass:error2m} replaced by Assumptions \ref{ass:errorcontrol1} and \ref{ass:error2control1}, respectively. Note that, $\eta_i$ does not need to be observable because it is integrated out, just as $\vec{A}_i$ is, when forming the identifying inequalities. By the same logic, the dimension of $\eta_i$ also does not affect the results.

\subsection{Time-Variant Controls.} 
Time-variant controls cannot be integrated out the way that $\vec{A}_i$ and the time-invariant controls are. We thus require these controls to be observable, and these variables will enter the identifying inequalities. Let $Z_{it}$ be  a vector of control variables. 
Suppose that instead of Assumption \ref{ass:error}, we have Assumption \ref{ass:errorcontrol2}:
\begin{assumption}\label{ass:errorcontrol2} \textup{(a)} The error term $\vec{\epsilon}_{it}$ is independent of $\vec{X}_{it}$ given $\vec{A}_i$ and $Z_{it}$,

\textup{(b)} $E[\vec{Y}_{it}|\vec{X}_{i1},\dots,\vec{X}_{iT},Z_{i1},\dots,Z_{iT},\vec{A}_i,\eta_i] = E[\vec{Y}_{it}|\vec{X}_{it},Z_{it},\vec{A}_i,\eta_i]$, for all $t$, and

\textup{(c)} the conditional c.d.f. of $\vec{\epsilon}_{it} $ given $\vec{A}_i$ and $Z_{it}$ is continuous everywhere.
\end{assumption}
Then, instead of Lemma \ref{lem:idineq}, we have Lemma \ref{lem:idineqcontrol}, the proof of which is given at the end of this section: 
\begin{lemma}\label{lem:idineqcontrol} Under Assumption \textup{\ref{ass:errorcontrol2}}, we have, for any positive integer $M\leq T$ and any cycle $t_1,t_2,\dots,t_M,$ $t_{M+1}=t_1$ in  $\{1,\dots, T\}$,
\begin{align}
\sum_{m=1}^ME[\vec{Y}_{it_m}'|\vec{X}_{it_1},\dots,\vec{X}_{it_M},Z_{it_1} =\dots=Z_{it_M}](\vec{X}_{it_{m}}'\beta - \vec{X}_{it_{m+1}}'\beta)\geq 0 \text{ almost surely}. \label{panelidcontrol}
\end{align}
\end{lemma}
As Lemma \ref{lem:idineqcontrol} shows, identification is based on the individuals for whom the control variable $Z_{it}$ does not vary across the time periods considered. 

For the point identification conditions under Assumption \ref{ass:errorcontrol2}, first redefine ${\cal G}$ and $G$. For any integer $M\geq 2$, any cycle $t_1,\dots,t_M,t_{M+1}= t_1$ in $\{1,\dots,T\}$, and any $\vec{x}_1,\dots,\vec{x}_{M} \in \vec{\cal X}$, let
\begin{equation}
g_{t_1,\dots,g_M}  (\vec{x}_1,\dots,\vec{x}_M)=\sum_{m=1}^M\{(\vec{x}_{m} - \vec{x}_{m+1})E[\vec{Y}_{it_m}|\vec{X}_{it_1}= \vec{x}_1,\dots,\vec{X}_{it_M}  = \vec{x}_M, Z_{it_1} =\dots=Z_{it_M}]\}.
\end{equation}
Let ${\cal G}_{t_1,\dots,t_M}$ be the support of $g_{t_1,\dots,t_M}(\vec{X}_{it_1},\dots,\vec{X}_{it_M})$, and let ${\cal G}_{M} = \cup_{t_1,\dots,t_M\in\{1,\dots,T\}}{\cal G}_{t_1,\dots,t_M}$. Let
\begin{equation}
{\cal G} = \cup_{M=2,\dots,T}{\cal G}_{M}.
\end{equation}
For any $k=1,\dots,K$ and any element $\vec{x}^{-k}$ in $\vec{X}_{it}^{-k}$, let $G_{s,t}^k(\vec{x}^{-k})$ be the conditional support of $X_{it}^k - X_{is}^k$ given that $\vec{X}_{it}^{-k} = \vec{X}_{is}^{-k} = \vec{x}^{-k}$ and that $Z_{it} = Z_{is}$. Let
\begin{align}
G = \cup_{s,t=1,\dots,T}\cup_{k=1,\dots,K}\cup_{\vec{x}^{-k}\in\vec{X}_{-k}}G_{s,t}^k(\vec{x}^{-k}). 
\end{align}
Also, suppose that we replace Assumption \ref{ass:error2m} with the following assumption.
\begin{assumption}\label{ass:error2control2}
For every $k=1,\dots,K$, with positive probability, the conditional support of $\vec{\epsilon}_{it}$ given $\vec{A}_i$ and $Z_{i1},\dots, Z_{iT}$ is $R^K$. 
\end{assumption}
Then Theorems \ref{thm:setid}, \ref{thm:sufbdm}, \ref{thm:suflgm}, and \ref{thm:nec} go through with the redefined ${\cal G}$ and with Assumptions \ref{ass:error} and \ref{ass:error2m} replaced by Assumptions \ref{ass:errorcontrol2} and \ref{ass:error2control2}, respectively. This is because the proofs of these theorems go through without further change once the redefined ${\cal G}$ and $G$ are used and the assumptions are replaced.

\begin{proof}[Proof of Lemma \textup{\ref{lem:idineqcontrol}}] Under Assumption \ref{ass:errorcontrol2}(a), Lemma \ref{lem:gradient} implies that $\vec{p}(\vec{v},\vec{a}, z) $ is cyclic monotone in $\vec{v}$ for any fixed $\vec{a}$ and fixed $z$, where
\begin{equation}
\vec{p}(\vec{v},\vec{a},z) = E[\vec{Y}_{it}|\vec{X}_{it}'\beta = \vec{v},\vec{A}_i=\vec{a}, Z_{it}=z].
\end{equation}
This implies that for any integer $M$,  any cycle $\vec{x}_1, \vec{x}_2,\dots,\vec{x}_M, \vec{x}_{M+1}= \vec{x}_1$ in $\vec{\cal X}$, any $\vec{a}$ and any $z$, we have
\begin{equation}
\sum_{m=1}^M \vec{p}(\vec{x}_{m}'\beta,\vec{a},z)'(\vec{x}_{m}'\beta - \vec{x}_{m-1}'\beta)\geq 0.\label{preineq} 
\end{equation}
Using Assumption \ref{ass:errorcontrol2}(b), we get, for any cycle $t_1,t_2,\dots,t_M, t_{M+1}=t_1$ in $\{1,\dots,T\}$, and all $(\vec{x}_1,\dots,\vec{x}_M, \vec{a},z)$, 
\begin{align}
\vec{p}(\vec{x}_m'\beta,\vec{a},z) &= E[\vec{Y}_{it_m}|\vec{X}_{it} = \vec{x}_m,\vec{A}_i=\vec{a}, Z_{it_m}=z]\nonumber\\
& = E[\vec{Y}_{it_m}|\vec{X}_{it_1} = \vec{x}_1,\dots,\vec{X}_{it_M}= \vec{x}_M,\vec{A}_i=\vec{a}, Z_{it_1}=\dots=Z_{it_M}=z]
\end{align}
Thus, (\ref{preineq}) implies that for any cycle $t_1,t_2,\dots,t_M, t_{M+1}=t_1$ in $\{1,\dots,T\}$, and all $(\vec{x}_1,\dots,\allowbreak\vec{x}_M, \vec{a},z)$,
\begin{align}
\sum_{m=1}^M E[\vec{Y}_{it_m}|\vec{X}_{it_1},\dots,\vec{X}_{it_M},\vec{A}_i, Z_{it_1},Z_{it_1}=\dots=Z_{it_M}] (\vec{X}_{it_{m}}'\beta - \vec{X}_{t_{m+1}}'\beta)\geq 0, \text{ almost surely}.
\end{align}
Integrating out $\vec{A}_i$ and $Z_{it_1}$, we get
\begin{align}
\sum_{m=1}^M E[\vec{Y}_{it_m}|\vec{X}_{it_1},\dots,\vec{X}_{it_M},  Z_{it_1}=\dots=Z_{it_M}] (\vec{X}_{it_{m}}'\beta - \vec{X}_{t_{m+1}}'\beta)\geq 0, \text{ almost surely}.\label{control2id}
\end{align}
This proves the lemma.
\end{proof}
\medskip

\end{document}